\documentclass[10pt]{article}
\usepackage{geometry}
\usepackage{graphicx}
\usepackage{amssymb}
\usepackage{amsmath}
\usepackage{alg}

\usepackage{pgf,tikz}
\usetikzlibrary{arrows,automata}

\usepackage{amsthm}
\theoremstyle{definition}
\newtheorem{thm}{Theorem}
\newtheorem{dfn}[thm]{Definition}

\newtheorem{prop}[thm]{Proposition}
\newtheorem{obs}[thm]{Observation}
\newtheorem{ex}[thm]{Example}

\newcommand{\setcomp}[2]{\ensuremath{\{#1\ |\ #2\}}}
\newcommand{\set}[1]{\ensuremath{\{#1\}}}

\newcommand{\kleene}[0]{\ensuremath{^*}}

\newcommand{\rnarrow}{\rightarrow}

\newcommand{\commentout}[1]{}

\title{Functional Logic Programming\\with Generalized Circular Coinduction\thanks{I would like to thank Prof.~Dr.~Horst Reichel for his useful input and comments.}}
\author{Ronald de Haan\\[5pt]\textit{Technische Universit\"at Dresden}}
\date{}

\usepackage[]{hyperref}

\usepackage[disable,colorinlistoftodos]{todonotes}

\begin{document}
\maketitle

\begin{abstract}
We propose a method to adapt functional logic programming to deal with reasoning on coinductively interpreted programs as well as on inductively interpreted programs.
In order to do so, we consider a class of objects interesting for this coinductive interpretation, namely regular terms.
We show how the usual data structures can be adapted to capture these objects.
We adapt the operational semantics of Curry to interpret programs coinductively.
We illustrate this method with several examples that show the working of our method and several cases in which it could be useful.
Finally, we suggest how the declarative semantics can be adapted suitably.
\end{abstract}

\section{Introduction}

Within the paradigm of declarative programming, several styles of programming have been developed, among which logic and functional programming. One of the benefits of logic programming is its use of free variables and its inherent capabilities to model search problems. The benefits of functional programming include its efficient execution. Attempts to combine the best of both worlds have resulted in functional logic programming. An example of this combination is the programming language Curry \cite{Hanus:1997p1, Hanus:2007p193}).

Another development within the paradigm of declarative programming is to develop mechanisms to reason about coinductively defined semantics, besides the usual inductively defined semantics. This is useful, for instance, to reason about infinite data, and processes and their states. In the case of logic programming, this resulted in co-logic programming \cite{Simon:2007p225}. Another approach to reason about coinductive semantics concerns proof mechanisms based on circularity \cite{Rosu:2000p309}, which are used in the coinductive prove engine CIRC \cite{Lucanu:2009tc}, for instance. However, in the case of functional logic programming, no such adaptation to a coinductive interpretation has been developed.

In this paper, we will generalize the idea of circular coinduction to be used in the domain of functional logic programming. Similarly to detecting circularity in equality proofs (as is the case in circular coinduction), we will detect circularity in a reduction derivation. By means of assuming a certain resulting value for terms that appear repeatedly in such circular reduction sequences, we can resolve these circularities. Circular coinduction is a particular (Boolean) case of this general mechanism. Hence the name generalized circular coinduction for this method.

In Section~\ref{sec:prelim}, we will define some preliminary notions used in our approach. We define a data structure to capture a particular (finitely representable) class of infinite terms, and we show how to perform several basic operations on these data structures. Then, in Section~\ref{sec:compstrat}, we describe how we propose to modify the computational strategy from \cite{Hanus:1997p1} with generalized circular coinduction. In Section~\ref{sec:examples}, we will show the benefits of this approach by means of several examples. We suggest a possible declarative semantics for the mechanism in Section~\ref{sec:declsem}, before concluding and suggesting directions of further research in Section~\ref{sec:conclusions}.

\section{Preliminaries}\label{sec:prelim}
\subsection{Cyclic terms}
Certain classes of infinite terms can be described finitely. One such class is the class of regular infinite terms. A term is regular if is has only finitely many different subterms. It is easy to see that all finite terms are regular. In this paper we will consider (finite and infinite) finitely-branching regular terms, such as the terms in Example~\ref{ex:regulartrees}. Example~\ref{ex:infinitelybranchingtree} shows a term that is not finitely-branching.

\begin{ex}\label{ex:regulartrees}
An example of a non-regular infinite term is the infinite tree $\Gamma_{1}$, where $\Gamma_{i}$ for $i \in \mathbb{N}$ is defined as follows, for a (fixed) binary constructor symbol $\star$:
\[ \Gamma_{i} = \star(i,\Gamma_{i+1}) \]
The set of subterms $sub(\Gamma_{1}) = \mathbb{N} \cup \setcomp{\Gamma_{j}}{j \in \mathbb{N}}$ of $\Gamma_{1}$ is infinite.
An example of a regular infinite term is the infinite tree $\Delta$, where $\Delta$ is defined as follows, for a (fixed) binary constructor symbol $\star$:
\[ \Delta = \star(1,\Delta) \]
The set of subterms $sub(\Delta) = \set{1, \Delta}$ of $\Delta$ is finite.
\end{ex}

\begin{ex}\label{ex:infinitelybranchingtree}
Another example of a regular infinite term is the infinite tree $\Lambda$, defined as follows, for a (fixed) constructor symbol $\circ$ of infinite arity.
\[ \Lambda = \circ(\Lambda,\Lambda,\ldots) \]
In other words, each of the (infinitely many) direct subterms of the root of $\Lambda$ is $\Lambda$. Note that the set of subterms $sub(\Lambda) = \set{\Lambda}$ of $\Lambda$ is finite. In this paper, however, we restrict ourselves to finitely branching terms.
\end{ex}

In the following we develop the technical machinery to represent such (finitely-branching) regular terms. We will use two different kinds of variables: \textit{frozen} variables and regular variables. We fix a set $\mathcal{Y} = \set{y,y',y'',\ldots} \cup \set{y_1, y_2, \ldots}$ of frozen variables and a set $\mathcal{X} = \set{x,x',x'',\ldots} \cup \set{x_1,x_2,\ldots}$ of regular variables.
We let $\Sigma = \mathcal{C} \ \dot\cup\ \mathcal{F}$ denote a signature consisting of constructor symbols $\mathcal{C}$ and function symbols $\mathcal{F}$. We let $Var : \Sigma(\mathcal{X} \cup \mathcal{Y}) \rightarrow \mathcal{P}(\mathcal{X})$ denote the function that returns the set of all  variables occurring in a given term. Also, we let $FVar : \Sigma(\mathcal{X} \cup \mathcal{Y}) \rightarrow \mathcal{P}(\mathcal{Y})$ denote the function that returns the set of all frozen variables occurring in a given term. We denote the extension of $Var$ (resp. $FVar$) to sets of terms also with $Var$ (resp. $FVar$).

Consider a mapping $\rho$ from a set $S$ to the same set $S$, that maps all but a finite set of values $\set{s_1,\ldots,s_n} \subseteq S$ to themselves.
In the following, we will interchangeably use the following three views of such mappings.
Either we view $\rho$ as a finite set of pairs $\set{s_1\mapsto t_1,\ldots,s_n\mapsto t_n}$ that correspond to the mappings of elements that are not mapped to themselves,
or we view $\rho$ as a total function with domain $\set{s_1,\ldots,s_n}$,
or we view $\rho$ as a total function with domain $S$, that maps all values not in $\set{s_1,\ldots,s_n}$ to themselves.
Also, for any such mapping $\rho$ and any $S \subseteq dom(\rho)$, we let $\rho|_{S} = \setcomp{s \mapsto t \in \rho}{s \in S}$.

We represent finitely branching, regular terms with the following structures.

\begin{dfn}[Cyclic terms]
A cyclic term over a signature $\Sigma$ and a set of variables $\mathcal{X}$ is a term $t \in \Sigma(\mathcal{X} \cup \mathcal{Y})$, together with a complete mapping $\theta : FVar(t) \rightarrow \Sigma(\mathcal{X} \cup FVar(t)) \backslash FVar(t)$.
We say that a cyclic term is in \textit{base form} if $Var(rng(\theta)) \subseteq Var(t)$.

A cyclic term consisting of $t$ and $\theta$ is denoted with $(t,\theta)$. We denote the set of all cyclic terms over the signature $\Sigma$, the set of variables $\mathcal{X}$ and the set of frozen variables $\mathcal{Y}$ with $CT(\Sigma,\mathcal{X},\mathcal{Y})$, or simply $CT$ if $\Sigma$, $\mathcal{X}$ and $\mathcal{Y}$ are clear from the context.
\end{dfn}

Example~\ref{ex:infinitelist} illustrates how finitely-branching regular terms can be represented by cyclic terms.

\begin{ex}\label{ex:infinitelist}
The term $\Delta$ from Example~\ref{ex:regulartrees} with $\star$ being the list constructor operator $:$ -- in other words, the infinite list containing only 1's -- can be represented as the cyclic term $(y, \set{y \mapsto 1:y})$.
\end{ex}

Certain cyclic terms that are syntactically different can intuitively represent the same terms. Below, we define a notion of equivalence of cyclic terms (based on bisimulation on labelled graphs extracted from cyclic terms) to capture this intuition. In Example~\ref{ex:equivalence}, we give an example of two syntactically different, but equivalent cyclic terms.

\begin{dfn}[Cyclic term graph]
For a given cyclic term $(t,\theta)$ we define the cyclic term graph $(V,E,VL,EL)$ as a graph $(V,E)$ together with a vertex labelling $VL : V \rightarrow 2^{\mathcal{Y}}$ and an edge labelling $EL : E \rightarrow \mathbb{N}$ as follows:
\begin{itemize}
  \item $V := (sub(t) \cup sub(rng(\theta)))\ \backslash\ \mathcal{Y}$;
  \item $y \in VL(v)$ iff $y \in \mathcal{Y}, v \in V$ and $y \mapsto v \in \theta$; and
  \item $(v,v') \in E$ and $EL(v,v') = n$ iff
  \begin{itemize}
    \item either $v'$ is the $n$-th subterm of $v$,
    \item or there exists a $y \in \mathcal{Y}$ such that $y$ is the $n$-th subterm of $v$ and $v'$ is the unique vertex in $V$ such that $y \in VL(v')$.
  \end{itemize}
\end{itemize}
Here $sub$ denotes the function that returns the set of subterms occurring in a given term.
Note that for every $y \in FVar(t)$ there exists a unique node $v \in V$ such that $y \in VL(v)$, namely the unique term $v \in \Sigma(\mathcal{X} \cup \mathcal{Y})$ for which holds $y \mapsto v \in \theta$.
\end{dfn}

\begin{dfn}[Cyclic term graph bisimulation]
For two given cyclic term graphs $(V,E,VL,EL)$ for the cyclic term $(t,\theta)$, and $(V',E',VL',EL')$ for the cyclic term $(t',\theta')$, we say that two nodes $w \in V$ and $w' \in V'$ bisimulate if:
\begin{itemize}
  \item there exists a relation $Z \subseteq V \times V'$ that satisfies the following conditions:
  \begin{itemize}
    \item if $(v,v') \in Z$, then the terms $v$ and $v'$ must have the same outermost symbol from $\Sigma$, or be the same variable from $\mathcal{X}$;
    \item if $(v_1,v'_1) \in Z$, and also $(v_1,v_2) \in E$ and $EL(v_1,v_2) = n$, then there must exist a $v'_2 \in V'$ such that $(v_2,v'_2) \in Z$, $(v'_1,v'_2) \in E'$ and $EL'(v'_1,v'_2) = n$; and
    \item if $(v_1,v'_1) \in Z$, and also $(v'_1,v'_2) \in E'$ and $EL'(v'_1,v'_2) = n$, then there must exist a $v_2 \in V$ such that $(v_2,v'_2) \in Z$, $(v_1,v_2) \in E$ and $EL(v_1,v_2) = n$;
  \end{itemize}
  \item $(w,w') \in Z$.
\end{itemize}
Note that this notion of bisimulation does not depend in any way on the vertex labelings $V$ and $V'$.

For any $y \in FVar(t)$ and $y' \in FVar(t')$, we say that $y$ and $y'$ bisimulate in the cyclic term graphs if we have that $u$ and $u'$ bisimulate in the cyclic term graphs, where $u \in V$ is the unique node such that $y \in VL(u)$ and $u' \in V'$ is the unique node such that $y' \in VL'(u')$.
\end{dfn}

\begin{dfn}[Equivalence]
Two cyclic terms $(t,\theta)$ and $(t',\theta')$ are equivalent, denoted with $(t,\theta) \equiv (t',\theta')$, iff $t$ and $t'$ bisimulate in the cyclic term graphs for $(t,\theta)$ and $(t',\theta)$.
\end{dfn}

\begin{ex}\label{ex:equivalence}
Consider the two cyclic terms $(t,\theta) = (y,\set{y \mapsto 1:y})$ and $(t',\theta') = (y',\set{y' \mapsto 1:1:y'})$. The cyclic term graph $(V,E,VL,EL)$ for $(t,\theta)$ is given by:
\[ V = \set{1, 1:y}, \quad E = \set{(1:y,1), (1:y,1:y)} \]
\[ VL(1) = \emptyset, \quad VL(1:y) = \set{y} \]
\[ EL(1:y,1) = 1, \quad EL(1:y,1:y) = 2 \]
The cyclic term graph $(V',E',VL',EL')$ for $(t',\theta')$ is given by:
\[ V' = \set{1,1:y',1:1:y'}, \quad E' = \set{(1:y',1), (1:y',1:1:y'), (1:1:y',1), (1:1:y',1:y')},\]
\[ VL'(1) = \emptyset, \quad VL'(1:y') = \emptyset, \quad VL'(1:1:y') = \set{y'} \]
\[ EL'(1:y',1) = 1, \quad EL'(1:1:y',1) = 1 \]
\[ EL'(1:y',1:1:y') = 2, \quad EL'(1:1:y',1:y') = 2 \]
The relation $Z$, given below, is a bisimulation that witnesses that $1:y \in V$ and $1:1:y' \in V'$ bisimulate (and also witnesses that $y$ and $y'$ bisimulate).
\[ Z = \set{(1,1), (1:y,1:y'), (1:y,1:1:y')} \]
This bisimulation thus also witnesses that $(t,\theta) \equiv (t',\theta')$.

The two cyclic term graphs, together with the given bisimulation are drawn in Figure~\ref{fig:cyclictermgraphs}.
\end{ex}

\begin{figure}[ht]
  \begin{center}
    \begin{tikzpicture}[->,>=stealth',shorten >=1pt,auto,node distance=2.8cm,semithick]
      \tikzstyle{every state}=[]
      \node[state] (a1) {$1\!\!:\!\!y$};
      \node[state] (a2) [below of=a1] {$1$};
      \path (a1) edge node {$1$} (a2);
      \path (a1) edge [loop left] node {$2$} (a1);
      
      \node (a1label) [above of=a1, node distance=.7cm] {$y$};
  
      \node[state] (b1) [right of=a1] {$1\!\!:\!\!1\!\!:\!\!y'$};
      \node[state] (b2) [right of=b1] {$1\!\!:\!\!y'$};
      \node[state] (b3) [below of=b1] {$1$};
      \path (b1) edge node {$1$} (b3);
      \path (b2) edge node {$1$} (b3);
      \path (b1) edge [bend left=10] node {$2$} (b2);
      \path (b2) edge [bend left=10] node {$2$} (b1);
      
      \node (b1label) [above of=b1, node distance=.9cm] {$y'$};
      
      \path [draw,-,dashed] (a2) -- (b3);
      \path [draw,-,dashed] (a1) -- (b1);
      \path [draw,-,dashed] (a1) edge [bend left=50] (b2);
    \end{tikzpicture}
  \end{center}
  \caption{Example of cyclic term graphs for $(y,\set{y \mapsto 1:y})$ and $(y',\set{y' \mapsto 1:1:y'})$, and a bisimulation. The bisimulation $Z$ is drawn with a dashed line.}
  \label{fig:cyclictermgraphs}
\end{figure}
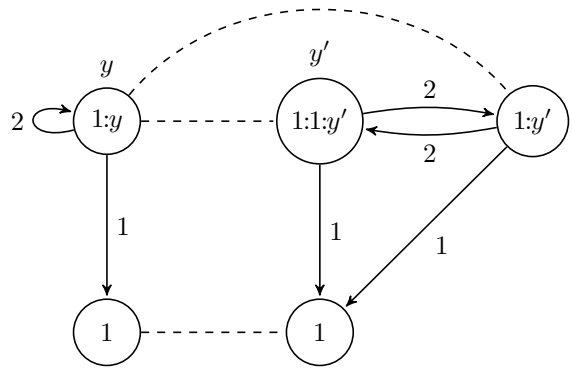

The notion of equivalence of cyclic terms allows us to state the following property. This property will allow us to assume, without loss of generality, that cyclic terms are in base form.

\begin{prop}\label{prop:baseformnormalization}
For every cyclic term there is an equivalent cyclic term in base form.
\end{prop}
\begin{proof}
Let $(t,\theta)$ be a cyclic term. We show that $(t',\theta') = (\theta(t),\theta|_{FVar(\theta(t))})$ is an equivalent cyclic term in base form.
Clearly, $(t',\theta')$ is a cyclic term.

We show that $(t',\theta')$ is equivalent to $(t,\theta)$. Consider the cyclic term graphs $(V,E,VL,EL)$ for $(t,\theta)$ and $(V',E',VL',EL')$ for $(t',\theta')$. Since $(t,\theta)$ and $(t',\theta')$ contain the same subterms (that are not frozen variables), we know $V = V'$.
Also, the construction of edges and edge labels in the cyclic term graph for a cyclic term $(t,\theta)$ is invariant under application of $\theta$ to $t$.
Thus, both cyclic term graphs are isomorphic (when not taking into account vertex labels). Since the definition of bisimulation does not use vertex labels, we get that $id$ witnesses that the term graphs bisimulate. Thus $(t,\theta)$ and $(t',\theta')$ are equivalent.

Now, we show that $(t',\theta')$ is in base form. Take an arbitrary $x \in Var(rng(\theta'))$. Since $\theta' \subseteq \theta$, we know $x$ is in $Var(rng(\theta))$ as well. Thus there is a $y \mapsto t_{y} \in \theta$ such that $x \in Var(t_{y})$. By definition of cyclic terms, we know that $y \in FVar(t)$. Thus $x \in Var(\theta(t)) = Var(t')$.
\end{proof}


When dealing with terms containing variables, a notion of substitutions is required. We define substitutions on cyclic terms and their effect as follows.

\begin{dfn}[Substitutions]
A substitution for cyclic terms is a mapping $\sigma: \mathcal{X} \rightarrow \Sigma(\mathcal{X} \cup \mathcal{Y})$ that maps all but finitely many variables to themselves, together with a total mapping $\omega : FVar(rng(\sigma)) \rightarrow \Sigma(\mathcal{X} \cup FVar(rng(\sigma))) \backslash FVar(rng(\sigma))$. Substitutions are denoted with $(\sigma,\omega)$. We denote the unique homomorphic extension $\hat\sigma$ of $\sigma$ on $\Sigma(\mathcal{X} \cup \mathcal{Y})$ also with $\sigma$ if no confusion arises. Similarly for the unique homomorphic extension $\hat\omega$ of $\omega$ on $\Sigma(\mathcal{X} \cup \mathcal{Y})$.
\end{dfn}

\begin{dfn}[Safe substitutions]
A substitution $(\sigma,\omega)$ is safe for a cyclic term $(t,\theta)$ if $dom(\omega) \cap FVar(t) = \emptyset$.
In other words, a substitution is safe for a cyclic term, if it doesn't redefine (the mapping of) frozen variables that were already defined.
\end{dfn}

\begin{dfn}[Proper substitutions]
A substitution $(\sigma,\omega)$ is proper for a cyclic term $(t,\theta)$ if $FVar(rng(\omega|_{FVar(\sigma(t))})) \subseteq FVar(\sigma(t))$.
In other words, a substitution is proper for a cyclic term $(t,\theta)$, if the the mappings of all frozen variables introduced by instantiation of a variable in $t$ contain no frozen variables that are not introduced by instantiation of a variable in$t$
\end{dfn}

\begin{dfn}[Effect of safe, proper substitutions]\label{dfn:effect}
The result $(\sigma,\omega)(t,\theta)$ of a substitution $(\sigma,\omega)$, safe and proper for a cyclic term $(t,\theta)$, to this cyclic term $(t,\theta)$ is defined as $(\sigma(t),(\sigma(\theta) \cup \omega)|_{FVar(\sigma(t))})$.
Here $\sigma(\set{y_1 \mapsto t_1,\ldots,y_n \mapsto t_n}) = \set{y_1 \mapsto \sigma(t_1), \ldots y_n \mapsto \sigma(t_n)}$.
\end{dfn}

\begin{prop}
The result $(t',\theta')$ of a substitution $(\sigma,\omega)$, safe and proper for a cyclic term $(t,\theta)$, to this cyclic term $(t,\theta)$ is a cyclic term.
\end{prop}
\begin{proof}
We show that the result $(t',\theta')$ is a cyclic term.
Assume without loss of generality that $(t,\theta)$ is in base form.
Since $(\sigma,\omega)$ is safe for $(t,\theta)$, we know that $dom(\sigma(\theta)) \cap dom(\omega) = \emptyset$. Therefore $\sigma(\theta) \cup \omega$ is a (functional) mapping, and thus so is $t'$.

Clearly, $dom(\theta') \subseteq FVar(t')$. We show that $FVar(t') \subseteq dom(\theta')$. Take an arbitrary $y \in FVar(t')$. Then either (i) $y \in FVar(t)$ or (ii) for some $x \in Var(t)$, we have $x \mapsto t_{y} \in \sigma$ for some $t_{y}$ and $y \in FVar(t_{y})$. In case (i), $y \in FVar(\sigma(t))$ and $y \in dom(\theta)$. So also, $y \in dom(\sigma(\theta))$. Therefore, $y \in dom(\theta')$. In case (ii), $y \in FVar(\sigma(t))$ and $y \in dom(\omega)$, thus $y \in dom(\theta')$.

Finally, we show that $FVar(rng(\theta')) \subseteq FVar(t')$. Take an arbitrary $y \in FVar(rng(\theta'))$. Then either (i) $y \in FVar(rng(\sigma(\theta)))$ or (ii) $y \in FVar(rng(\omega))$. In case (i), either (i.a) $y \in FVar(rng(\theta))$ or (i.b) for some $x \in Var(rng(\theta))$, we have $x \mapsto t_{y} \in \sigma$ and $y \in FVar(t_{y})$.
In case (i.a), $y \in FVar(t)$, so $y \in FVar(t')$.
In case (i.b), since $(t,\theta)$ is in base form, $x \in Var(t)$ and thus $y \in FVar(t')$.
In case (ii), we know that $y \in FVar(rng(\omega|_{FVar(\sigma(t))}))$. Then, since $(\sigma,\omega)$ is proper for $(t,\theta)$, we know that $y \in FVar(t')$.

We showed that $(t',\theta')$ has the property that $\theta'$ is a total mapping from $FVar(t')$ to $\Sigma(\mathcal{X} \cup FVar(t')) \backslash FVar(t')$. Thus $(t',\theta')$ is a cyclic term.
\end{proof}

These notions of equivalence and substitutions, defined above, allow us to define a notion of unification.

\begin{dfn}[Unification]
A substitution $(\sigma,\omega)$ unifies two cyclic terms $(t,\theta)$ and $(t',\theta')$ if 
\begin{itemize}
  \item it is safe and proper for both $(t,\theta)$ and $(t',\theta')$, and
  \item $(\sigma,\omega)(t,\theta) \equiv (\sigma,\omega)(t',\theta')$.
\end{itemize}
\end{dfn}

In order to avoid unwanted side-effects of operations on cyclic terms (or substitutions) that happen to share (frozen) variables, we define a method to standardize the frozen variables in two cyclic terms apart. We introduce frozen variable renamings, and use them to standardize the frozen variables in cyclic terms and substitutions apart.

\begin{dfn}[Frozen variable renaming]
A frozen variable renaming is a function $f: \mathcal{Y} \rightarrow \mathcal{Y}$ that maps finitely many $y \in \mathcal{Y}$ to different and pairwise disjoint $y' \in \mathcal{Y}$, and all other frozen variables to themselves. We denote the unique homomorphic extension $\hat f$ of a frozen variable renaming $f$ on $\Sigma(\mathcal{X} \cup \mathcal{Y})$ also with $f$, if no confusion arises.
\end{dfn}

\begin{dfn}[Safe frozen variable renaming]
A frozen variable renaming $f = \set{y_1 \mapsto y'_1, \ldots, y_n \mapsto y'_n}$ is safe for a cyclic term $(t,\theta)$ if for all $1 \leq i \leq n$ we have that $y'_i \not\in FVar(t)$.

A frozen variable renaming $f = \set{y_1 \mapsto y'_1, \ldots, y_n \mapsto y'_n}$ is safe for a substitution $(\sigma,\omega)$ if for all $1 \leq i \leq n$ we have that $y'_i \not\in dom(\omega)$.
\end{dfn}

\begin{dfn}[Renaming frozen variables in cyclic terms]
The result $f(t,\theta)$ of a frozen variable renaming $f$, safe for a cyclic term $(t,\theta)$, to this cyclic term $(t,\theta)$ is the cyclic term $(f(t),f(\theta))$, where $f(\theta) = \setcomp{f(y) \mapsto f(t_{y})}{y \mapsto t_{y} \in \theta}$. The fact that $f$ is safe for $(t,\theta)$, together with the fact that the range of $f$ is pairwise disjoint, assures that $f(t,\theta)$ is a cyclic term, and thus that applying a safe frozen variable renaming to a cyclic term results in a cyclic term.
\end{dfn}

\begin{obs}
For any cyclic term $(t,\theta)$ and any frozen variable renaming $f$, safe for $(t,\theta)$, we have that $(t,\theta) \equiv f(t,\theta)$.
\end{obs}
\begin{proof}[Proof (sketch).]
There exists a trivial bisimulation between the cyclic term graphs of $(t,\theta)$ and $f(t,\theta)$.
\end{proof}

\begin{dfn}[Renaming frozen variables in substitutions]
The result $f(\sigma,\omega)$ of a frozen variable renaming $f$ to the substitution $(\sigma,\omega)$ is the substitution $(f(\sigma),f(\omega))$.
The fact that $f$ is safe for $(\sigma,\omega)$, together with the fact that the range of $f$ is pairwise disjoint, ensures that the result is a substitution.
\end{dfn}

\begin{prop}
Let $(t,\theta)$ be a cyclic term, and let $(\sigma,\omega)$ be a substitution, safe and proper for $(t,\theta)$. Also, let $f$ be a frozen variable renaming safe for $(\sigma,\omega)$. If $f(\sigma,\omega)$ is safe and proper for $(t,\theta)$, then we have that $(f(\sigma,\omega))(t,\theta) \equiv (\sigma,\omega)(t,\theta)$.
\end{prop}
\begin{proof}[Proof (sketch).]
Since $f$ has an effect only on frozen variables, and since the construction of cyclic term graphs is invariant under (uniform) substitution of frozen variables, the cyclic term graphs for $(f(\sigma,\omega))(t,\theta)$ and $(\sigma,\omega)(t,\theta)$ are identical. Thus $id$ is the bisimulation that witnesses that $(f(\sigma,\omega))(t,\theta) \equiv (\sigma,\omega)(t,\theta)$.
\end{proof}

\begin{dfn}[Standardization apart]
Two cyclic terms $(t,\theta)$ and $(t',\theta')$ are said to have their frozen variables standardized apart if $FVar(t) \cap FVar(t') = \emptyset$.
Two substitutions $(\sigma,\omega)$ and $(\sigma',\omega')$ are said to have their frozen variables standardized apart if $dom(\omega) \cap dom(\omega') = \emptyset$.
\end{dfn}

Cyclic terms and substitutions for cyclic terms can be standardized apart by means of safe variable renaming. In the following, we will assume that cyclic terms and substitutions have their frozen variables standardized apart. With this method of standardization apart in place, we can return to substitutions and define composition of substitutions (that have their frozen variables standardized apart).

\begin{dfn}[Composition of substitutions]
Let $(\sigma,\omega)$ and $(\sigma',\omega')$ be two substitutions, with frozen variables standardized apart. We define the composition of these two substitutions to be $(\sigma,\omega)\circ(\sigma',\omega') = (\sigma \circ \sigma', \omega \cup \omega')$.
The fact that the two substitutions have their frozen variables standardized apart ensures that the result is a substitution.
\end{dfn}

\begin{prop}
Let $(\sigma,\omega)$ and $(\sigma',\omega')$ be two substitutions, that have their frozen variables standardized apart. Let $(\sigma',\omega')$ be safe and proper for a given cyclic term $(t,\theta)$, and let $(\sigma,\omega)$ be safe and proper for $(\sigma',\omega')(t,\theta)$. Then $(\sigma,\omega) \circ (\sigma',\omega')$ is safe and proper for $(t,\theta)$.
\end{prop}
\begin{proof}
Let $(\sigma'',\omega'') = (\sigma,\omega) \circ (\sigma',\omega') = (\sigma \circ \sigma', \omega \cup \omega')$. Since $(\sigma,\omega)$ and $(\sigma',\omega')$ have their frozen variables standardized apart, we know $(\sigma'',\omega'')$ is a substitution.

We show that $(\sigma'',\omega'')$ is safe for $(t,\theta)$. Since $(\sigma',\omega')$ is safe for $(t,\theta)$, we know that $dom(\omega') \cap FVar(t) = \emptyset$. Also, since $(\sigma,\omega)$ is safe for $(\sigma',\omega')(t,\theta)$, we know that $dom(\omega) \cap FVar(\sigma'(t)) = \emptyset$. This implies that $dom(\omega) \cap FVar(t) = \emptyset$. Thus $(dom(\omega) \cup dom(\omega')) \cap FVar(t) = \emptyset$. And thus $dom(\omega\cup\omega') \cap FVar(t) = \emptyset$. Thus $(\sigma'',\omega'')$ is safe for $(t,\theta)$.

We show that $(\sigma'',\omega'')$ is proper for $(t,\theta)$.
Take an arbitrary $y \in FVar(rng(\omega''|_{FVar(\sigma''(t))}))$. We show that $y \in FVar(\sigma''(t))$.
We know that there exists a $y' \mapsto t_{y'} \in \omega''$ such that $y' \in FVar(\sigma''(t))$ and $y \in FVar(t_{y'})$.
Since $y' \in FVar(\sigma(\sigma'(t)))$, we know that either (i) $y' \in FVar(\sigma'(t))$ or (ii) for some $x \in Var(\sigma'(t))$ we have that $y' \in FVar(\sigma(x))$.
In case (i), since $(\sigma',\omega')$ is proper for $(t,\theta)$, we know that $y \in FVar(\sigma'(t))$, and thus also $y \in FVar(\sigma''(t))$.
In case (ii), since $(\sigma,\omega)$ is proper for $(\sigma',\omega')(t,\theta)$, we know $y \in FVar(\sigma(\sigma'(t))) = FVar(\sigma''(t))$.
\end{proof}

Finally, we introduce some notational conventions. Firstly, we will often denote cyclic terms $(t,\theta)$ with $t$, and denote substitutions for cyclic terms $(\sigma,\omega)$ with $\sigma$. Secondly, we introduce a notation to refer to subterms of cyclic terms. In order to do so, we define the notion \textit{decomposable form} for cyclic terms.

\begin{dfn}
A cyclic term $(t,\theta)$ is in decomposable form if $t \not\in \mathcal{Y}$ and for every direct subterm $t' \in sub(t)$ of $t$ holds that $FVar(rng(\theta|_{FVar(t')}) \subseteq FVar(t')$.
\end{dfn}
\begin{prop}
For every cyclic term $(t,\theta)$, there is an equivalent cyclic term $(t',\theta')$ in decomposable form.
\end{prop}
\begin{proof}[Proof (sketch)]
For any cyclic term $(t,\theta)$ and any mapping $y \mapsto t_{y} \in \theta$, applying the mapping $\rho = \set{y \mapsto t_{y}}$ to any subterm of $t$ or $rng(\theta)$ results in a cyclic term $(t',\theta')$ equivalent to $(t,\theta)$ (after restricting the domain of $\theta'$ if necessary). By applying a finite number of such equivalence preserving transformations, any cyclic term can be transformed into decomposable form.
\end{proof}

In the following, we assume without loss of generality that cyclic terms are in decomposable form. Any cyclic terms that are not in decomposable form, we implicitly transform to decomposable form.

\begin{dfn}
We say that a cyclic term $(t,\theta)$ is of the form $d(t_1,\ldots,t_n)$, for $d \in \mathcal{C} \cup \mathcal{F}$, if
$t = d(t_1,\ldots,t_n)$. We assume without loss of generality that $(t,\theta)$ is in decomposable form.
In this case, for $1 \leq i \leq n$, we let $(t_i,\theta_i)$ denote the cyclic term that is obtained by letting $\theta_i = \theta|_{FVar(t_i)}$.
\end{dfn}

Thirdly, we define substitution of subterms of cyclic terms.

\begin{dfn}
For a given cyclic term $(t,\theta)$ we define the (possibly infinite) set of positions $pos(t,\theta) \subseteq \mathbb{N}\kleene$ as follows.
\[ pos(x,\theta) = \set{\epsilon} \]
\[ pos(d(t_1,\ldots,t_n),\theta) = \set{\epsilon} \cup \setcomp{1\cdot \overline{m}}{\overline{m} \in pos(t_1)} \cup \cdots \cup \setcomp{n\cdot \overline{m}}{\overline{m} \in pos(t_n)} \]
Note that we do not explicitly handle the case for $t \in \mathcal{Y}$, since we implicitly transform cyclic terms to decomposable form.
\end{dfn}

\begin{dfn}[Fixed-position substitutions]
For given cyclic terms $(t,\theta)$ and $(t',\theta')$ (with their frozen variables standardized apart), and a position $\overline{m} \in pos(t,\theta)$ of $(t,\theta)$, we define the fixed-position substitution $(t,\theta)[(t',\theta')]_{\overline{m}}$ as:
\[ \begin{array}{r l l}
(t,\theta)[(t',\theta')]_{\overline{m}} &=& (t[t']_{\overline{m}},(\theta \cup \theta')|_{t[t']_{\overline{m}}}) \\
t[t']_{\epsilon} &=& t' \\
d(t_1,\ldots,t_n)[t']_{i\cdot \overline{m}'} &=& d(t_1,\ldots,t_{i-1},t_{i}[t']_{\overline{m}'},t_{i+1},\ldots,t_n) \\
\end{array} \]
Again, we do not explicitly handle the case for $t \in \mathcal{Y}$, since we implicitly transform cyclic terms to decomposable form. Note that since the two cyclic terms $(t,\theta)$ and $(t',\theta')$ are in decomposable form and have their frozen variables standardized apart, the result $(t,\theta)[(t',\theta')]_{\overline{m}}$ of a fixed-position substitution is a cyclic term.
\end{dfn}

Fourthly, in the following, we will denote cyclic terms $(t,\theta)$ where $\theta = \set{y_1 \mapsto t_1, \ldots, y_n \mapsto t_n}$ in Curry notation as \verb|t {y1 -> t1, ..., yn -> tn}|.

Note that in this paper, we will not make the connection between cyclic terms in the above syntactic sense and regular terms precise. The intuition is that the set of cyclic terms corresponds to the set of (finitely-branching) regular terms. Making this intuition precise is a topic of further research, and is closely related to the topic of finding a declarative semantics of our approach (see Section~\ref{sec:declsem}).

\subsubsection{Adding typing}
For the purposes of functional logic programming, we would like to extend the above machinery with the concept of types. For the sake of simplicity of presentation, we have not done this explicitly above. Extending the above with typing is, however, very straightforward. We describe the steps that have to be taken to extend the above with typing.

We need a many sorted signature $\Sigma$, to start with. Also, we explicitly partition the set $\mathcal{X}$ of variables and the set $\mathcal{Y}$ of frozen variables into sets for each sort. Then, for cyclic terms $(t,\theta)$, we pose the usual typing constraints on $t$ and on all mappings in $\theta$. For substitutions $(\sigma,\omega)$ on cyclic terms, we pose the usual typing constraints on all mappings in $\sigma$ and $\omega$. Finally, we straightforwardly add typing constraints on frozen variable renamings and fixed-position substitutions.

\subsubsection{Computing unification}
An algorithm to compute whether a unifier for two cyclic terms (that have their frozen variables standardized apart) exists, and to compute such a unifier, is given in Algorithm~\ref{alg:unify}. This algorithm is based on the Martelli-Montanari algorithm \cite{Martelli:1982ty}.

\begin{thm}[Termination]\label{thm:termination}
The unification algorithm always terminates, no matter what choices are made.
\end{thm}
\begin{proof}[Proof (sketch)]
Similarly to the proof of termination for the Martelli-Montanari algorithm \cite{Martelli:1982ty}, we define a well-founded ordering that decreases with every action taken. To the ordering used in \cite{Martelli:1982ty}, we add (as most significant dimension), the following measure.
\[ maxsize(H) - size(H) \]
We know $\theta$ and $\theta'$ are fixed and finite.
Since only finitely many variables from $\mathcal{X}$ occur in $E$ (and the number of such variables never increases), the number of equations in $F$ is finitely bounded. For every substitution $x \mapsto t$ applied to terms in $E$, the size of $t$ is bounded by the size of $E$ (at the time of applying $x \mapsto t$). For every substitution $y \mapsto t \in F$, the size of $t$ is also bounded by the size of $E$ (at the time of adding $y \mapsto t$ to $F$).
For every substitution $y \mapsto t \in \theta \cup \theta'$, the size of $t$ is bounded as well.
There are finitely many different possible mappings that can be applied to equations in $E$. Furthermore, there exists an upper bound on this number of mappings. Thus, since $E$ is initially finite, we know there is an upper bound for the number of different equations that can occur in $E$. Since only equations occurring in $E$ can be added to $H$, the size of $H$ also has this bound. Thus $maxsize(H)$ is well-defined, and therefore this is a valid measure.

It is easy to see that with each action, this well-founded ordering decreases. Thus the algorithm terminates.
\end{proof}

\begin{thm}[Partial correctness]\label{thm:partialcorrectness}
On termination of the algorithm, a unifying substitution is returned iff a unifying substitution exists.
\end{thm}
\begin{proof}[Proof (sketch)]
It is easy to see that if an $(E,F)$ is returned, then it is a substitution. If $E$ contains anything other than equations of the form $x \doteq t$, where $x$ occurs only once in $E$, then there are still actions that can be performed. Also, for any mapping $y \mapsto t \in F$, we know that $y$ occurs in the rhs of some mapping in $E$, because of the structure of the only rule introducing mappings to $F$. Because of the structure of this same rule, we have that for any $y$ that occurs in the rhs of a mapping in $E$, there is a mapping $y \mapsto t \in F$. Furthermore, because of the structure of this same rule, we know that $(E,F)$ is safe and proper for both $(t,\theta)$ and $(t',\theta')$.

In order to see that $(E,F)$ is a unifier of $(t,\theta)$ and $(t',\theta')$, consider the following condition, in which we fix suitable mappings $\sigma$ and $\omega$.
\begin{center}
There exists a bisimulation between the cyclic term graphs\\
of $(\sigma(t),(\theta \cup \theta' \cup \omega)|_{\sigma(t)})$ and $(\sigma(t'),(\theta \cup \theta' \cup \omega)|_{\sigma(t')})$, for each $t \doteq t' \in E \cup H$.
\end{center}

This condition holds on successful termination of the algorithm, for $(\sigma,\omega) = (E,F)$. For all equations $x \doteq t_{x} \in E$, this holds since $x \mapsto t_{x} \in \sigma$. For equations $y \doteq t_{y} \in H$ the condition holds as well. We know $y \doteq t_{y}$ can only be in $H$ if it was in $E$ before. If there is no such bisimulation (as specified in the condition), by the structure of the rules, the algorithm would return with $\bot$. However, since the algorithm terminated successfully (by assumption), we know such a bisimulation must exist.

We know that the algorithm can only halt with failure because of rule (2). By the structure of rule (2), we know if the algorithm halts with failure, the condition does not hold (for any $\sigma$). It suffices to show, for each (other) action that can be performed, that the condition holds before applying the action iff it holds after applying the action. The cases for rules (1), (3), (4) and (5) are trivial. The cases for rules (6), (7), (8) and (9) can be proven straightforwardly. This proves that the condition holds on initiating the algorithm, for $(\sigma,\omega) = (E,F)$, if the algorithm halts with $(E,F)$. It also proves that the condition holds for no $(\sigma,\omega)$ on initiating the algorithm if the algorithm halts with $\bot$.
From this claim, it follows that the algorithm halts with a unifier $(E,F)$ iff such a unifier exists.
\end{proof}

\begin{algorithm}[h!]
  \caption{Unification algorithm for cyclic terms.}
  \vspace{10pt}
  \algname{unify}{$(t,\theta),(t',\theta')$}
  \alginout{Two cyclic terms $(t,\theta)$ and $(t',\theta')$, with frozen variables standardized apart.}
                  {A unifier $(\sigma,\omega)$ of the two cyclic terms if it exists, or else $\bot$.}
  \vspace{10pt}
  
  \begin{algtab} 
  \end{algtab}
  
  Let $E = \set{t \doteq t'}$. Let $H = \emptyset$. Let $F = \emptyset$. While some action can be performed, and no failure has occurred, do the following. Nondeterministically choose from the set $E$ of equations an equation of a form below (such that the corresponding condition holds) and perform the corresponding action. Note that $x$ ranges over variables in $\mathcal{X}$, $y$ over frozen variables in $\mathcal{Y}$, and $t$ over terms in $\Sigma(\mathcal{X} \cup \mathcal{Y})$.
  
  \vspace{10pt}
  \begin{small}
  \begin{tabular}{p{.5cm} p{4cm} p{2.2cm} p{5cm}}
  &\textbf{form}&\textbf{condition}&\textbf{action}\\
  \hline
  (1)&$f(s_1,\ldots,s_n) \doteq f(t_1,\ldots,t_n)$&$\top$&replace by the equations $s_1 \doteq t_1,$ $\ldots,$ $s_n \doteq t_n$\\
  (2)&$f(s_1,\ldots,s_n) \doteq g(t_1,\ldots,t_n)$&$f \neq g$&halt with failure\\
  (3)&$x \doteq x$&$\top$&delete the equation\\
  (4)&$t \doteq x$&$t \not\in \mathcal{X}$&replace by the equation $x \doteq t$\\
  (5)&$t \doteq y$&$t \not\in \mathcal{X} \cup \mathcal{Y}$&replace by the equation $y \doteq t$\\
  (6)&$x \doteq t$&$x \not\in Var(t)$ and $x$ occurs in another equation&perform the substitution $x \mapsto t$ on all other equations\\
  (7)&$x \doteq t$&$x \in Var(t)$&for a fresh $y$: replace by the equation $x \doteq y$; apply the substitution $x \mapsto y$, on all other equations; and add $y \mapsto t\set{x \mapsto y}$ to $F$\\
  (8)&$y \doteq t$&$y \doteq t \in H$&delete the equation\\
  (9)&$y \doteq t$&$y \doteq t \not\in H$&add $y \doteq t$ to $H$; and apply the unique substitution $y \mapsto t'$ from $\theta \cup \theta' \cup F$ to $t$\\
  \end{tabular}
  \end{small}

  When the algorithm terminates with failure, return $\bot$. Otherwise, return $(E,F)$, where $E$ is seen as a mapping from $\mathcal{X}$ to $\Sigma(\mathcal{X} \cup \mathcal{Y})$.

  \label{alg:unify}
\end{algorithm}

\subsection{Systems of equations}

Another possible representation of regular (infinite) terms would be systems of equations. Such representations of infinite terms can, for instance, be found in (certain implementations of) Prolog and several languages used for coinductive specifications. Because of their use in different contexts, the reader might be familiar with systems of equations, and might appreciate a motivation for the choice of using cyclic terms instead of systems of equations. In fact, there is a one-to-one correspondence between cyclic terms and systems of equations. Below, we illustrate how cyclic terms can be straightforwardly transformed to systems of equations, and vice versa. A major advantage of the use of cyclic terms is that this manner of representing regular terms allows us to to define the operational semantics in such a way that the connection to the operational semantics of regular Curry is very clear (see Section~\ref{sec:compstrat}).

\subsubsection{Correspondence with cyclic terms}

We informally describe how systems of equations can be transformed straightforwardly to cyclic terms, and vice versa. We will do so by means of an example. A formal definition of systems of equations and the transformations to and from cyclic terms is beyond the scope of this paper.

\begin{ex}
Consider the regular infinite tree represented by the cyclic term, for symbols $f,g,h \in \Sigma$:
\[ (f(g(y,y')),\set{y \mapsto g(y,y'), y' \mapsto h(y')}) \]
This regular tree can also be represented by the following set of equations, with root $x$:
\[ \set{ x = f(g(x',x'')), \quad x' = g(x',x''), \quad x'' = h(x'') } \]
\end{ex}

It is easy to see that a cyclic term $(t,\theta)$ can be written as a system of equations by mapping each frozen variable to a (regular) variable, and (possibly) adding an extra variable as root to refer to $t$.

Conversely, a (finite) system of equations can be written as a cyclic term $(t,\theta)$ as follows.
We map the variables that occur on the left hand side of equations to (distinct) frozen variables.
We take the frozen variable corresponding to the distinguished variable in the equation as $t$.
We let $\theta$ be the mappings corresponding to the equations.
Possibly, we need to unfold $t$ by applying mappings from $\theta$ until all frozen variables occurring in $\theta$ also occur in $t$.

\section{Computational strategy}\label{sec:compstrat}
With the (well-behaved) data structures to represent regular terms in place, we are ready to define the computational strategy that uses generalized circular coinduction. Note that in the following, we use cyclic terms (often simply denoted terms) and the corresponding notions of substitutions and unifiers.

\subsection{Informally}\label{sec:compstratinformal}

The intuition behind our computational strategy is fairly simple. The idea is to detect circularity in the reductions performed, and breaking out of such a circle by assigning a possible value to the (function symbol rooted) term which is evaluated repeatedly. What the possible values are is something that needs to be specified by the programmer. The reason for this is discussed in Section~\ref{sec:restrictions}.

More concretely, this results in the following computational mechanism, based on the needed narrowing mechanisms from \cite{Hanus:1997p1, Antoy:2000p256}. Analogously to the method of detecting cyclic behavior in \cite{Simon:2006p226}, in every derivation, we keep track of the (function symbol rooted) expressions on which narrowing is applied in the derivation, together with the total term they appear in and their position in this total term. As soon as the term $t$ (appearing as subterm of the total term $t_{tot}$) on which narrowing is about to be applied, and some term $t'$ (appearing as subterm of the total term $t'_{tot}$) on which narrowing is applied previously, unify (with unifier $\sigma$), we do the following.
\begin{itemize}
  \item We (nondeterministically) guess a possible resulting value for the function symbol rooted expression occurring at the root of both $t$ and $t'$, from the possibilities specified by the programmer.
  \item To the derivation, we add the constraint $\sigma(t^r_{tot}) \doteq \sigma(t'^r_{tot})$, where $t^r_{tot}$ is the term $t_{tot}$ in which the subterm $t$ is replaced by the guessed result, and $t'^r_{tot}$ is defined similarly.
  \item We continue the derivation with $\sigma(t^r_{tot})$.
\end{itemize}

Also, if during the derivations we encounter constraints of the form $t \doteq t'$,
we try to unify them.

\subsection{Formal definition}
In order to formally define our computational mechanism, we recall some preliminary notions, known from the literature on functional logic programming and Curry. We define patterns and (partial) definitional trees.

\begin{dfn}[Taken from Definition~2 in \cite{Antoy:2000p256}]
A pattern is a term $f(t_1,\ldots,t_n)$, where $n$ is the arity of $f \in \mathcal{F}$ and each $t_i$, for $1 \leq i \leq n$, is a constructor term.
\end{dfn}

\begin{dfn}[Definition~12 in \cite{Antoy:2000p256}]
$\mathcal{T}$ is a partial definitional tree (pdt) with pattern $\pi$ iff one of the following cases holds:
\begin{itemize}
  \item $\mathcal{T} = branch(\pi,o,\mathcal{T}_1,\ldots,\mathcal{T}_k)$, where $\pi$ is a pattern, $o$ is the occurrence of a variable in $\pi$, the sort of $\pi|_o$ has constructors $c_1,\ldots,c_k$ for some $k>0$, and for all $i \in \set{1,\ldots,k}$, $\mathcal{T}_i$ is a pdt with pattern $\pi[c_i(X_1,\ldots,X_n)]_o$, where $n$ is the arity of $c_i$ and $X_1,\ldots,X_n$ are new distinct variables.
  \item $\mathcal{T} = leaf(\pi)$, where $\pi$ is a pattern.
\end{itemize}
\end{dfn}

\begin{dfn}[Taken from \cite{Antoy:2000p256} and \cite{Hanus:1997p1}]
Let $\mathcal{R}$ be a rewrite system. $\mathcal{T}$ is a definitional tree of an operation $f$ iff $\mathcal{T}$ is a pdt whose pattern argument is $f(X_1,\ldots,X_n)$, where $n$ is the arity of $f$ and $X_1,\ldots,X_n$ are new distinct variables, and for every rule $l \rightarrow r$ of $\mathcal{R}$ with $l = f(t_1,\ldots,t_n)$, there exists a leaf $leaf(\pi)$ of $\mathcal{T}$ such that $l$ is a variant of $\pi$, and we say that the node $leaf(\pi)$ represents the rule $l \rightarrow r$,
in which case we will also write $leaf(l \rightarrow r)$.
We write $pattern(\mathcal{T})$ for the pattern of a definitional tree $\mathcal{T}$ and $DT$ for the set of all definitional trees.
\end{dfn}

Besides the notions used in the usual computational mechanism of Curry, we will need a few additional notions to make our computational strategy work. We define memory configurations, and several types of expressions used in the computation.

\begin{dfn}
We denote with $Mem$ the set of all memory configurations $\mathcal{P}( CT \times CT \times \mathbb{N}\kleene )$.
A memory configuration thus is a set of triples consisting of two cyclic terms and a sequence of natural numbers (representing a position in the second cyclic term).
\end{dfn}

\begin{dfn}
We denote with $EQ$ the set of all equations of the form $t \doteq t'$, where $t,t' \in CT$, and the trivial (dummy) equation $\top$.
\end{dfn}

\begin{dfn}[Based on definitions in \cite{Hanus:1997p1}]
An answer expression is a pair $\langle \sigma, e \rangle$, consisting of a substitution $\sigma$ and an expression $e$. An answer expression $\langle \sigma, e \rangle$ is solved if $e$ is a constructor term.
A disjunctive expression is a (multi-)set of answer expressions $\set{\langle \sigma_1, e_1 \rangle, \ldots, \langle \sigma_n, e_n \rangle}$, sometimes written as $\langle \sigma_1,e_1 \rangle \vee \ldots \vee \langle \sigma_n,e_n \rangle$. The set of all disjunctive expressions is denoted with $\mathcal{D}$.
A memory answer expression is a pair $\langle \sigma, e, M \rangle$, where $\langle \sigma,e \rangle$ is an answer expression, and $M \in Mem$ is a memory configuration. A memory answer expression $\langle \sigma, e, M \rangle$ is solved if the answer expression $\langle \sigma, e \rangle$ is solved.
A disjunctive memory expression is a (multi-)set of memory answer expressions. Notation is analogous to the notation of disjunctive expressions. The set of all disjunctive memory expressions is denoted $\mathcal{D}_{Mem}$.
A memory constraint answer expression is a pair $\langle \sigma, e, M, eq \rangle$, where $\langle \sigma,e,M \rangle$ is a memory answer expression, and $eq \in EQ$ is an equation.
A disjunctive memory constraint expression is a (multi-)set of memory constraint answer expressions. Notation is analogous to the notation of disjunctive expressions. The set of all disjunctive memory constraint expressions is denoted $\mathcal{D}_{Mem,EQ}$.
\end{dfn}

In order to handle the additional types of expressions defined above, we adapt several auxiliary functions known from the usual computational strategy of Curry.

\begin{dfn}We define the function $pair: (\mathcal{D}_{Mem} \cup \set{suspend}) \times EQ \rightarrow \mathcal{D}_{Mem,EQ} \cup \set{suspend}$ as follows.
\[ pair(suspend,eq) = suspend \]
\[ pair(\set{\langle \sigma_1,e_1,M_1\rangle, \ldots, \langle \sigma_n,e_n,M_n \rangle},eq) = \set{\langle \sigma_1,e_1,M_1,eq\rangle,\ldots,\langle \sigma_n,e_n,M_n,eq\rangle} \]
\end{dfn}

\begin{dfn}[Based on definitions in \cite{Hanus:1997p1}] We define the auxiliary functions $compose$ and $replace$ as follows.
Note that the function $cst$ will be defined below.
\[ compose(t,\mathcal{T},\sigma,t',o',M) = \left \{ \begin{array}{l l}

\set{\langle\sigma,t,M\rangle}&\mbox{if $cst(t,\mathcal{T},t',o',M) =$}\\
&\mbox{$suspend$}\\

\\

\{\langle\sigma_1 \circ \sigma,t_1,M_1\rangle, \ldots,&
\mbox{if $cst(t,\mathcal{T},t',o',M) =$}\\
\langle \sigma_n \circ \sigma,t_n,M_n\rangle\}&
\mbox{$\{\langle\sigma_1,t_1,M_1\rangle, \ldots, \langle\sigma_n,t_n,M_n\rangle\}$}\\

\end{array} \right . \]
\[ replace(t,o,suspend) = suspend \]
\[ replace(t,o,\set{\langle\sigma_1,t_1,M_1,eq_1\rangle,\ldots,\langle\sigma_n,t_n,M_n,eq_n\rangle}) = \]\[ \set{\langle\sigma_1,\sigma_1(t)[t_1]_o,M_1,eq_1\rangle,\ldots,\langle\sigma_n,\sigma_1(t)[t_n]_o,M_n,eq_n\rangle} \]

\end{dfn}

With all the preliminary notions in place, we can now define formally what possible guesses for circular behaviour are.

\begin{dfn}
Possible guesses are defined by a partial function $\rho : CT \rightarrow \mathcal{P}(CT)$. We say that $t'$ is a possible guess for a cyclic term $t \in CT$ if $t \in dom(\rho)$ and $t' \in \rho(t)$.
\end{dfn}

We are now ready to define the two functions $cs$ and $cst$ that handle the computational steps in derivations.

\begin{dfn}[Based on definitions in \cite{Hanus:1997p1}]\label{dfn:cs} The function $cs : CT \times CT \times \mathbb{N}\kleene \times Mem \rightarrow \mathcal{D}_{Mem,EQ} \cup \set{suspend}$, designed to perform a computation step, is defined as follows.
\begin{small}
\[ \begin{array}{r l l l r}

cs(x,t_{tot},o,M)&=&suspend&\mbox{for all variables $x \in \mathcal{X}$}&(1)\\
\\
cs(t \doteq t',t_{tot},o,M)&=&\set{\langle \sigma, \top, M, \top \rangle}&\mbox{if $\sigma$ unifies $t$ and $t'$}&(2)\\
\\
cs(f(t_1,\ldots,t_n),t_{tot},o,M)&=&pair( cst(f(t_1,\ldots,t_n),\mathcal{T},t_{tot},o,&\mbox{if $\mathcal{T}$ is a definitional}&\\
&&M),\top)&\mbox{tree for $f$ and rule (2)}&\\
&&&\mbox{does not apply}&(3.1)\\
&&&\\
&&\cup\\
\\
&&\{ \langle \sigma, g,M,\sigma(t_{tot})[g]_{o} \doteq \sigma(t'_{tot})[g]_{o'}\rangle&\mbox{where $t = f(t_1,\ldots,t_n)$}\\
&&|\ \langle t',t'_{tot},o'\rangle \in M, \mbox{$\sigma$ unifies $t$ and}&\\
&&\mbox{$t'$, $g$ is a possible guess for $\sigma(t)$} \}&&(3.2)\\
\\
cs(c(t_1,\ldots,t_n),t_{tot},o,M)&=&replace(c(t_1,\ldots,t_n),k,&\mbox{if $cs(t_i,t_{tot},o\cdot i,M) =$}\\
&&cs(t_k,t_{tot},o\cdot k,M))&\mbox{$suspend$, for all $1 \leq i < k$}\\
&&&\mbox{and $cs(t_k,t_{tot},o\cdot k,M) \neq$}&\\
&&&\mbox{$suspend$}&(4)\\
\\
cs(c(t_1,\ldots,t_n),t_{tot},o,M)&=&suspend&\mbox{if $cs(t_i,t_{tot},o\cdot i,M) =$}&\\
&&&\mbox{$suspend$ for all $1 \leq i \leq n$}&(5)\\
\end{array} \]
\end{small}
\end{dfn}

Notice the use of nondeterminism in Definition~\ref{dfn:cs}, in particular in case (3). In order to allow the programmer to use any form of nondeterminism, we opt for as much nondeterminism in the computational strategy as possible. The choice to break cyclicity in derivations (case (3.2)) or simply continue the derivation (case (3.1)) is nondeterministic. When breaking cyclicity, only the single currently evaluated term is replaced by a possible guessed value, not all occurrences of this term in the answer expression. Also, we do not remember what guessed values we have applied in a derivation. This way, different guesses can be applied to (several occurrences of) the same term in one derivation.

\begin{dfn}[Based on definitions in \cite{Hanus:1997p1}]
The function $cst: CT \times DT \times CT \times \mathbb{N}\kleene \times Mem \rightarrow (\mathcal{D}_{Mem} \cup \set{suspend})$, designed to perform the application of a reduction (narrowing) step, is defined as follows.
\begin{small}
\[ cst(t,leaf(l \rightarrow r),t',o',M) = \set{\langle id,\sigma(r), M \cup \set{\langle t,t',o'\rangle}\rangle} \quad \mbox{if $\sigma$ is a substitution with $\sigma(l) = t$} \]
\[ \begin{array}{r}
cst(t,branch(\pi,o,\mathcal{T}_1,\\
\ldots,\mathcal{T}_k),t',o',M)\\
\end{array} =
\left \{ \begin{array}{l l}

cst(t,\mathcal{T}_i,t',o',M)&\mbox{if $t|_o = c(t_1,\ldots,t_n)$ and}\\
&\mbox{$pattern(\mathcal{T}_i)|_o = c(X_1,\ldots,X_n)$}\\
\\
\emptyset&\mbox{if $t|_o = c(t_1,\ldots,t_n)$ and}\\
&\mbox{$pattern(\mathcal{T}_i)|_o \neq c(\ldots)$, $i = 1,\ldots,k$}\\
\\
\bigcup^k_{i=1} compose(\sigma_i(t),\mathcal{T}_i,\sigma_i,t',o',M)&\mbox{if $t|_o = X$ and}\\
&\mbox{$\sigma_i = \set{X \mapsto pattern(\mathcal{T}_i)|_o}$}\\
\\
replace(t,o,cs(t|_o,t',o'\cdot o,M))&\mbox{if $t|_o = f(t_1,\ldots,t_n)$}\\
\end{array} \right . \]

\end{small}
\end{dfn}

The main difference between the above definitions for the functions $cs$ and $cst$ and their definitions in \cite{Hanus:1997p1} is that in the above definitions we keep track of (1) a memory configuration $M$, (2) the total term $t_{tot}$ (resp. $t'$) on which the computation step takes place and (3) the position $o$ (resp. $o'$) in this total term that is currently being evaluated. Also, with each application of a reduction rule (the first case for $cst$), we add an appropriate entry to the memory configuration. This allows us to resolve the circular behavior in case (3.2) for $cs$, exactly as described in Section~\ref{sec:compstratinformal}.

We remember each term on which a reduction step is applied (the first case for $cst$). It would be sufficient to remember only those terms for which a possible guess is available. This would result in an equivalent (yet in principle more efficient) computational mechanism. For the sake of simplicity of presentation, we don't make this distinction.

Finally, we introduce a few standard reduction rules and we define what derivations are.
We use the following definitions of $\wedge$ and $\Rightarrow$, as in \cite{Hanus:1997p1}.
\[ \top \wedge X \rightarrow X \]
\[ X \wedge \top \rightarrow X \]
\[ (\top \Rightarrow X) \rightarrow X \]

\begin{dfn}[Based on definitions in  \cite{Hanus:1997p1}]
A rewriting step for a disjunctive memory expression (denoted with $\rightarrow$) goes as follows.
\[ D \vee \langle \sigma, e, M \rangle \vee D' \rnarrow D \vee \langle \sigma_1 \circ \sigma, eq_1 \Rightarrow e_1, M_1 \rangle \vee \ldots \vee \langle \sigma_n \circ \sigma, eq_n \Rightarrow e_n, M_n \rangle \vee D' \]
\[ \mbox{if $\langle \sigma, e, M\rangle$ is unsolved and $cs(e,e,\epsilon,M) = \set{\langle \sigma_1, e_1, M_1, eq_1 \rangle, \ldots, \langle \sigma_n, e_n, M_n, eq_n \rangle}$} \]
Here, we let $\top \Rightarrow t$ denote $t$.
A derivation is a finite sequence of rewriting steps. A derivation for a term $e$ is a derivation starting with the disjunctive memory expression $\langle \emptyset, e, \emptyset \rangle$.
\end{dfn}

\subsection{Restrictions on the assumptions}\label{sec:restrictions}

The need to specify explicitly what are possible resulting values for a term that is evaluated repeatedly in a circular derivation, is illustrated nicely by the following example. Consider the following program, given in Curry syntax.

\begin{verbatim}
and [] = True
and (x:xs) = x && (and xs)

trues = True:trues
\end{verbatim}

Without any restrictions on the resulting values that can be guessed in circular derivations, we would be able to derive both \verb|True| and \verb|False| from \verb|and trues| (by guessing \verb|True| and \verb|False|, respectively, as the resulting value of \verb|and trues|). However, if we want to interpret \verb|and| as conjunction over (finite and infinite) lists of Boolean values, only the answer \verb|True| would be appropriate. Also, if we want to interpret \verb|and| as the predicate denoting all finite lists of Boolean values whose conjunction is true, only the answer \verb|False| would be appropriate. Since the interpretation is to be determined by the programmer, we need to let the programmer specify such choices. The possible guesses can intuitively be considered as the possible meaning of cyclic behavior.

An advantage of placing explicit restrictions on the possible values for such nondeterministic steps in the derivation, is that the number of possibilities is reduced enormously. This obviously has a positive effect on the efficiency of computation.

How to specify the possible resulting values for each situation is a question that needs to be investigated further. With a program, the user needs to specify possible guesses with a partial function $\rho: CT \rightarrow \mathcal{P}(CT)$. A concise specification language in which every (reasonable) choice of values can be expressed is desirable. In this paper we will describe the possible guesses with rewrite rules of the form \verb|f| $\leadsto$ \verb|v| (where \verb|f| is a term, and \verb|v| a set of terms), that get the following interpretation. If a term unifies with \verb|f| with unifier $\theta$, the possible guesses of the value of \verb|f| are the values in the set $\theta$(\verb|v|). When \verb|v| contains the anonymous variable \verb|_|, by convention, we let the corresponding possible guess $\theta(\verb|_|)$ be a fresh variable.

\section{Examples}\label{sec:examples}

In order to illustrate the merit of our approach, we will discuss several examples. These examples show that certain problems can be programmed and solved in a very intuitive fashion in our approach. We will also argue that the problems exemplified here cannot be solved in a similarly intuitive manner in regular Curry. A number of derivations for the following examples can be found in Appendix~\ref{sec:examplederivations}.

\subsection{Infinite lists}

In the program $\mathcal{P}_1$ in Figure~\ref{fig:ex1}, we consider infinite lists. This example shows that it is possible to use the usual definitions of conjunction (\verb|and|), quantifiers (\verb|forall| and \verb|exists|) and concatenation (\verb|append|) to extend these functions to (regular) infinite lists, with the intended interpretation. Also, we can straightforwardly introduce a function \verb|inf| that decides whether infinitely many elements of a list satisfy a certain property.

\begin{figure}[h]
\begin{center}
\textbf{Reduction rules:}\\
\begin{verbatim}
and [] = True
and (x:xs) = x && (and xs)

map f [] = []
map f (x:xs) = (f x):(map f xs)

forall f xs = and (map f xs)
exists f xs = not (and (map (\x -> not (f x)) xs))

inf f xs = inf' f False (loop xs)
inf' f b (x:xs) = inf' f ((f x) || b) xs

loop (x:xs) = loop xs

append [] ys = ys
append (x:xs) ys = x:(append xs ys)

isEven n = mod n 2 == 0
isOdd n = not (isEven n)

ones = 1:ones
twos = 2:twos
\end{verbatim}

\vspace{5pt}\textbf{Assumption possibilities:}\\
\begin{tabular}{l l l}
\verb|and _| & $\leadsto$ & \verb|{True}| \\
\verb|inf' _ b _| & $\leadsto$ & \verb|{b}| \\
\verb|loop x| & $\leadsto$ & \verb|{x}| \\
\verb|append xs _| & $\leadsto$ & \verb|{_}| \\
\end{tabular}
\end{center}
\caption{Example program $\mathcal{P}_1$.}
\label{fig:ex1}
\end{figure}

We can make the following derivations by using program $\mathcal{P}_1$. These example derivations illustrate that the operations \verb|and|, \verb|forall|, \verb|exists|, \verb|inf| and \verb|append| get their intended interpretation on infinite lists.

\vspace{10pt}\begin{tabular}{l l l l}
\textbf{term}&&\textbf{value}&\textbf{substitution}\\
\verb|forall isOdd ones| &$\rightarrow\kleene$& \verb|True|\\
\verb|exists isOdd twos| &$\rightarrow\kleene$& \verb|False|\\
\verb|inf isOdd (2:ones)| &$\rightarrow\kleene$& \verb|True|\\
\verb|inf isEven (2:ones)| &$\rightarrow\kleene$& \verb|False| \\
\verb|append ones [2,3]| &$\rightarrow\kleene$& \verb|y {y -> 1:y}|\\
\end{tabular}\vspace{10pt}

Regular Curry also supports infinite lists up to a certain extent, by means of lazy evaluation. We can, for instance, consider the infinite list \verb|ones|, and reason with such infinite lists. However, reasoning on infinite lists in regular Curry is restricted to reasoning over (unbounded) finite sublists. For instance, for a conjunction over an infinite list of Boolean values, only cases in which the resulting value is \verb|False| will terminate. In our approach, reasoning over (regular) infinite lists is not restricted in such a fashion. Of course, the reasoning used in our approach could be implemented in regular Curry as well, but this would require additional data types and reasoning methods to be made explicit. This would prevent a natural way of modelling.

\subsection{B\"uchi automata}

In the program $\mathcal{P}_2$ in Figure~\ref{fig:ex2}, we consider B\"uchi automata, and in particular the example automaton given below. Note that the definitions for \verb|inf| and \verb|loop| are exactly the same as in $\mathcal{P}_1$ in Figure~\ref{fig:ex1}. This example shows that we can directly encode the definition of a B\"uchi automaton, and the acceptance conditions of B\"uchi automata, resulting in a mechanism to decide acceptance (of regular infinite lists) of the particular automaton.

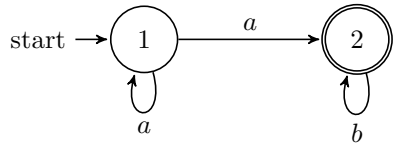
\vspace{5pt}
\begin{center}
\begin{tikzpicture}[->,>=stealth',shorten >=1pt,auto,node distance=2.8cm,semithick]
  \tikzstyle{every state}=[]
  \node[state,initial] (1) {$1$};
  \node[state,accepting] (2) [right of=1] {$2$};
  \path (1) edge [loop below] node {$a$} (1);
  \path (2) edge [loop below] node {$b$} (2);
  \path (1) edge [] node {$a$} (2);
\end{tikzpicture}
\end{center}

\begin{figure}[h]
\begin{center}
\textbf{Reduction rules:}\\
\begin{verbatim}
-- particular automaton
trans 1 a = 1
trans 1 a = 2
trans 2 b = 2

initial = 1

final = 2

-- mechanism
path q (s:ss) = q:(path (trans q s) ss)
accept string = inf (\x -> x == final) (path initial string)

inf f xs = inf' f False (loop xs)
inf' f b (x:xs) = inf' f ((f x) || b) xs

loop (x:xs) = loop xs
\end{verbatim}

\vspace{5pt}\textbf{Assumption possibilities:}\\
\begin{tabular}{l l l}
\verb|inf' _ b _| & $\leadsto$ & \verb|{b}| \\
\verb|loop x| & $\leadsto$ & \verb|{x}| \\
\end{tabular}

\end{center}
\caption{Example program $\mathcal{P}_2$.}
\label{fig:ex2}
\end{figure}

We can make the following derivations by using program $\mathcal{P}_2$. These example derivations show that the direct encoding of B\"uchi automata in $\mathcal{P}_2$ lead to a direct implementation that decides acceptance of infinite sequences. Note that this example shows that we can directly use nondeterminism (the B\"uchi automaton encoded in $\mathcal{P}_2$ is nondeterministic).

\vspace{10pt}\begin{tabular}{l l l l}
\textbf{term}&&\textbf{value}&\textbf{substitution}\\
\verb|accept (a:a:y) {y -> b:y}| &$\rightarrow\kleene$& \verb|True|\\
\verb|accept x| &$\rightarrow\kleene$& \verb|True| & \verb|{x -> a:y {y -> b:y}}|\\
\end{tabular}\vspace{10pt}

Of course, it is also possible to encode B\"uchi automata in regular Curry. However, encoding the acceptance conditions is not as straightforward as in our approach. In that case, namely, one has to explicitly represent (cyclic) paths and the existence of a final state in cycles of these paths.

\subsection{Bisimulation in Kripke structures}

In the program $\mathcal{P}_3$ in Figure~\ref{fig:ex3}, we consider Kripke structures, and in particular the example structure given below. Note that the definitions for \verb|and|, \verb|forall| and \verb|exists| are exactly the same as in $\mathcal{P}_1$ in Figure~\ref{fig:ex1}. This example shows that we can directly encode the definition of several Kripke structures, and the bisimulation conditions of Kripke structures, resulting in a mechanism to decide bisimulation of the particular Kripke structures.

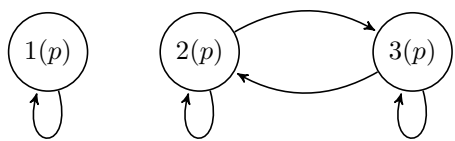
\vspace{5pt}
\begin{center}
\begin{tikzpicture}[->,>=stealth',shorten >=1pt,auto,node distance=2.8cm,semithick]
  \tikzstyle{every state}=[]
  \node[state] (1) {$1 (p)$};
  \path (1) edge [loop below] node {$$} (1);
\end{tikzpicture}
\hspace{20pt}
\begin{tikzpicture}[->,>=stealth',shorten >=1pt,auto,node distance=2.8cm,semithick]
  \tikzstyle{every state}=[]
  \node[state] (2) {$2 (p)$};
  \node[state] (3) [right of=2] {$3 (p)$};
  \path (2) edge [loop below] node {$$} (2);
  \path (3) edge [loop below] node {$$} (3);
  \path (2) edge [bend left] node {$$} (3);
  \path (3) edge [bend left] node {$$} (2);
\end{tikzpicture}
\end{center}

\begin{figure}[h]
\begin{center}
\textbf{Reduction rules:}\\
\begin{verbatim}
-- kripke structures
state "m1" = 1
trans "m1" 1 = 1
val "m1" 1 = 'p'

state "m2" = 2
state "m2" = 3
trans "m2" 2 = 2
trans "m2" 2 = 3
trans "m2" 3 = 2
trans "m2" 3 = 3
val "m2" 2 = 'p'
val "m2" 3 = 'p'

-- mechanism
bisim m1 w1 m2 w2 =
  let next1 = findall (trans m1 w1)
  and next2 = findall (trans m2 w2) in
    sameSet (findall (val m1 w1)) (findall (val m2 w2)) &&
    forall (\v1 -> (exists (\v2 -> bisim m1 v1 m2 v2) next2) next1) &&
    forall (\v2 -> (exists (\v1 -> bisim m1 v1 m2 v2) next1) next2)

and [] = True
and (x:xs) = x && (and xs)

forall f xs = and (map f xs)
exists f xs = not (and (map (\x -> not (f x)) xs))

sameSet xs ys = (subSet xs ys) && (subSet ys xs)
subSet [] _ = True
subSet (x:xs) ys = (elem x ys) && subSet xs ys
\end{verbatim}

\vspace{5pt}\textbf{Assumption possibilities:}\\
\begin{tabular}{l l l}
\verb|and _| & $\leadsto$ & \verb|{True}| \\
\verb|bisim _ _ _ _| & $\leadsto$ & \verb|{True}| \\
\end{tabular}
\end{center}
\caption{Example program $\mathcal{P}_3$.}
\label{fig:ex3}
\end{figure}

We can make the following derivations by using program $\mathcal{P}_3$. These example derivations illustrate that the direct encoding of the bisimulation conditions directly lead to a decision procedure.

\vspace{10pt}\begin{tabular}{l l l l}
\textbf{term}&&\textbf{value}&\textbf{substitution}\\
\verb|bisim "m1" 1 "m2" 2| &$\rightarrow\kleene$& \verb|True|\\
\verb|bisim "m1" 1 "m2" x| &$\rightarrow\kleene$& \verb|True| & \verb|{x -> 3}|\\
\end{tabular}\vspace{10pt}

Similarly to the example of B\"uchi automata given above, it is possible to encode bisimulation of Kripke structures in regular Curry, but not nearly as directly and straightforwardly as in our approach.

\subsection{Manipulation of infinite lists}

In the program $\mathcal{P}_4$ in Figure~\ref{fig:ex4}, we consider manipulation of infinite lists. This example shows that we can directly use the definition of several list manipulation operators from the finite case also in the infinite case (with regular terms), with the intended interpretation.

\begin{figure}[h]
\begin{center}
\textbf{Reduction rules:}\\
\begin{verbatim}
zip [] ys = ys
zip (x:xs) (y:ys) = x:y:(zip xs ys)

odd [] = []
odd (x:xs) = x:(even xs)

even [] = []
even (x:xs) = odd xs

ones = 1:ones
twos = 2:twos

nat n = n:(nat (n+1))
\end{verbatim}

\vspace{5pt}\textbf{Assumption possibilities:}\\
\begin{tabular}{l l l}
\verb|zip _ _| & $\leadsto$ & \verb|{_}| \\
\verb|odd _| & $\leadsto$ & \verb|{_}| \\
\end{tabular}
\end{center}
\caption{Example program $\mathcal{P}_4$.}
\label{fig:ex4}
\end{figure}

We can make the following derivations by using program $\mathcal{P}_4$. These examples show that the operations \verb|zip|, \verb|even| and \verb|odd| can be naturally used for infinite lists as well as finite lists. An equally natural encoding of such operations on infinite lists in regular Curry is not possible.

\vspace{10pt}\begin{tabular}{l l l l}
\textbf{term}&&\textbf{value}&\textbf{substitution}\\
\verb|zip ones twos| &$\rightarrow\kleene$& \verb|1:2:y {y -> 1:2:y}|\\
\verb|zip ones ones| &$\rightarrow\kleene$& \verb|1:1:y {y -> 1:1:y}|\\
\verb|odd (zip ones x)| &$\rightarrow\kleene$& \verb|1:y {y -> 1:y}|&(see below)\\
\verb|even (zip x twos)| &$\rightarrow\kleene$& \verb|2:y {y -> 2:y}|&(see below)\\
\verb|odd ones| &$\rightarrow\kleene$& \verb|1:y {y -> 1:y}|\\
\verb|even ones| &$\rightarrow\kleene$& \verb|1:y {y -> 1:y}|\\
\end{tabular}\vspace{10pt}

No derivation for \verb|odd (zip ones x)| and \verb|even (zip x twos)| leading to the mentioned values has an empty substitution. Each derivation results in a substitution mapping \verb|x| to an infinite list (representable by a cyclic term). Furthermore, for every substitution mapping \verb|x| to any cyclic term representing an infinite list (containing only fresh variables), there is a derivation leading to the mentioned values with this substitution.

Note that the term \verb|nat n|, for any integer \verb|n|, is infinite and non-regular. Our computational strategy is not suited to reason about such terms. This example illustrates this. We have that evaluating the following terms results in a diverging derivation. Similar effects occur when such non-regular infinite terms are used in the previous examples.

\vspace{10pt}\begin{tabular}{l l l}
\textbf{term}\\
\verb|odd (zip ones (nat 1))| &$\rightarrow\kleene$& \ldots\\
\verb|even (zip (nat 1) twos)| &$\rightarrow\kleene$& \ldots\\
\end{tabular}

\section{Declarative semantics}\label{sec:declsem}

Naturally, when modifying the operational semantics to interpret programs coinductively, we would like to change the denotational, or declarative, semantics accordingly. In this section, we suggest a possibility for suitable denotational semantics. Also, by means of several examples we illustrate how this suggested semantics differs from the inductive case. Furthermore, these examples serve to illustrate the suitability of the suggested semantics for the coinductive case.

For inductively interpreted functional logic programs, initial algebra denotational semantics is well-suited. Dually, for coinductively interpreted functional logic programs, we suggest a final coalgebra denotational semantics. For a general background on algebrae and coalgebrae, see for instance \cite{Jacobs:1997p238}. Consider the following (partial) signature specifying a particular type \textit{List}:
\[ List \rightarrow []\ |\ \mathbb{N} : List\ ; \]
Intuitively, in the inductive case, terms (of type $List$) correspond to finite lists of natural numbers. In the coinductive case, intuitively, cyclic terms (of type $List$) correspond to finite or infinite lists of natural numbers.

A suitable denotational semantics for terms of type $List$ is the initial $F$-algebra on the category of sets $\textbf{Set}$, where the corresponding functor $F : \set{\bot} \cup (\mathbb{N} \times X) \rightarrow X$ is derived directly from the signature. Call this initial algebra $I\!A$.\footnote{There are more such initial $F$-algebrae on $\textbf{Set}$, but they are all isomorphic.} Terms correspond to elements of $I\!A$. In fact, the set of finite lists on $\mathbb{N}$ is a suitable initial algebra. Now, functions from terms to natural numbers get the denotation of an $F$-algebra on $\mathbb{N}$, functions from terms to terms get the denotation of an $F$-algebra on $I\!A$. For instance, the (Curry encoding of the) function $length : List \rightarrow \mathbb{N}$ returning the length of a list would denote the algebra $(\mathbb{N},\alpha)$, given by:
\[ \alpha : \set{\bot} \cup (\mathbb{N} \times \mathbb{N}) \rightarrow \mathbb{N} \]
\[ \bot \mapsto 0 \quad (x,n) \mapsto 1 + n \]
By initiality of $I\!A$, there is exactly one morphism from $I\!A$ to this algebra on $\mathbb{N}$, which coincides with the function returning the length of lists in $I\!A$. Also, for instance, pairs of terms of type $List$ are assigned the denotation of elements in the algebra $(I\!A \times I\!A)$. In an analogous fashion, denotational semantics in algebraic terms can be assigned to the complete program.

The suggested final coalgebra semantics for cyclic terms of type $List$ are completely dual to the initial algebra semantics for the inductive case. In this semantics, cyclic terms denote elements of the final $G$-coalgebra on $\textbf{Set}$, where the corresponding functor $G : X \rightarrow \set{\bot} \cup (\mathbb{N} \times X)$ is derived directly from the signature. Call this final coalgebra $FC$.\footnote{There are more such final $G$-coalgebrae on $\textbf{Set}$, but they are all isomorphic.} Cyclic terms correspond to elements of $FC$. In fact, the set of all finite and infinite lists on $\mathbb{N}$ is a suitable final coalgebra. Assigning denotations to functions works dually to the inductive case. Functions from natural numbers to cyclic terms are $G$-coalgebrae on $\mathbb{N}$, and functions from cyclic terms to cyclic terms are $G$-coalgebae on $FC$. For instance, the function \verb|repeat| from natural numbers to cyclic terms given by \verb|repeat n = (n::(repeat n))| would denote the coalgebra $(\mathbb{N},\beta)$, given by:
\[ \beta : \mathbb{N} \rightarrow \set{\bot} \cup (\mathbb{N} \times \mathbb{N}) \]
\[ n \mapsto (n,n) \]
By finality of $FC$, there is exactly one morphism from this coalgebra on $\mathbb{N}$ to $FC$, which coincides with the function mapping any natural number to the infinite list containing only this number. Also, for instance, pairs of cyclic terms are assigned the denotation of elements in the coalgebra $(FC \times FC)$. Again, analogously, the whole program can be assigned denotational semantics in coalgebraic terms.

The above exposition is a gross oversimplification, of course. Things get more intricate, for instance, when a program is interpreted partially inductively and partially coinductively. Further research is needed on the topic of denotational semantics. We merely suggest a direction for research in this area.

\section{Conclusions}\label{sec:conclusions}
We showed how functional logic programming can be adapted to interpret programs coinductively as well as inductively. We singled out a particular class of (possibly infinite) objects interesting for this kind of reasoning, namely regular terms, and showed how the usual data structures can be modified to capture these objects. We showed how the operational semantics of the functional logic programming language Curry can be altered to allow for coinductive reasoning, and suggested how a suitable declarative semantics can be obtained. Furthermore, we illustrated the working and usefulness of our methods with several examples.

Working out a declarative semantics in full detail, and relating this semantics to the operational semantics, would be a topic of further research. Another direction for further research would be to investigate whether and how the computational mechanisms used in this paper could be optimized.

\bibliographystyle{ieeetr}
\bibliography{bibliography}

\listoftodos

\pagebreak
\appendix
\section{Example derivations}\label{sec:examplederivations}
In Figure~\ref{fig:examplederivations}, we include several derivations for the example programs given in Section~\ref{sec:examples}. For each derivation, we indicate certain (meaningful) steps in the derivation, together with the cases of the function $cs$ used to get from the previous step to this step (indicated in the column titled \textbf{case of $cs$}), and the (partial) resulting substitution calculated in the derivation from the previous step to this step (indicated in the column titled $\sigma$).

\begin{figure}[h!]
\begin{footnotesize}
\begin{center}

\begin{tabular}{l l l l}
&\textbf{term}&\textbf{case(s) of $cs$}&$\sigma$\\
&\verb|forall isOdd ones|&\\
$\rightarrow$&\verb|and (map isOdd ones)|&(3.1)\\
$\rightarrow$&\verb|and ((isOdd 1):(map isOdd ones))|&(3.1)\\
$\rightarrow$&\verb|and (True:(map isOdd ones))|&(3.1)\\
$\rightarrow$&\verb|True && (and (map isOdd ones))|&(3.1)\\
$\rightarrow$&\verb|and (map isOdd ones)|&(3.1)\\
$\rightarrow\kleene$&\verb|True|&(3.2), (2), (3.1)\\
\\ \\
&\textbf{term}&\textbf{case(s) of $cs$}&$\sigma$\\
&\verb|append ones [2,3]|&\\
$\rightarrow$&\verb|1:(append ones [2,3])|&(3.1)\\
$\rightarrow$&\verb|(x1 == 1:x1) => x1|&(3.2)\\
$\rightarrow$&\verb|True => y1 {y1 -> 1:y1}|&(2)&\verb|{x1 -> y1 {y1 -> 1:y1}}|\\
$\rightarrow$&\verb|y1 {y1 -> 1:y1}|&(3.1)\\
\\ \\
&\textbf{term}&\textbf{case(s) of $cs$}&$\sigma$\\
&\verb|accept x|&\\
$\rightarrow$&\verb|inf (==final) (path initial x)|&(3.1)\\
$\rightarrow$&\verb|inf' (==final) False (loop (path initial x))|&(3.1)\\
$\rightarrow$&\verb|inf' (==final) False (loop (1:(path 2 x2)))|&(3.1)&\verb|{x -> (a:x2)}|\\
$\rightarrow$&\verb|inf' (==final) False (loop (path 2 x2))|&(3.1)&\\
$\rightarrow$&\verb|inf' (==final) False (loop (2:(path 2 x3)))|&(3.1)&\verb|{x2 -> (b:x3)}|\\
$\rightarrow$&\verb|inf' (==final) False (loop (path 2 x3))|&(3.1)&\\
$\rightarrow\kleene$&\verb|inf' (==final) False (path 2 y1) {y1 -> b:y1}|&(3.2), (2), (3.1)&\verb|{x3 -> y1 {y1 -> b:y1}}|\\
$\rightarrow$&\verb+inf' (==final) ((2==final) || False) (path 2 y1)+&(3.1)&\\
$\rightarrow\kleene$&\verb|inf' (==final) True (path 2 y1) {y1 -> b:y1}|&(3.1)&\\
$\rightarrow\kleene$&\verb|inf' (==final) True (path 2 y1) {y1 -> b:y1}|&(3.1)&\\
$\rightarrow\kleene$&\verb|True|&(3.2), (2), (3.1)&\\
\\ \\
&\textbf{term}&\textbf{case(s) of $cs$}&$\sigma$\\
&\verb|odd (zip ones x)|&\\
$\rightarrow$&\verb|odd (1:x1:(zip ones xs1))|&(3.1)&\verb|{x -> (x1:xs1)}|\\
$\rightarrow\kleene$&\verb|1:(odd (zip ones xs1))|&(3.1)&\\
$\rightarrow$&\verb|(x2 == 1:x2) => x2|&(3.2)&\verb|{x -> y1, xs1 -> y1|\\
&&&\verb|         {y1 -> x1:y1}}|\\
$\rightarrow$&\verb|True => y2 {y2 -> 1:y2}|&(2)&\verb|{x2 -> y2 {y2 -> 1:y2}}|\\
$\rightarrow$&\verb|y2 {y2 -> 1:y2}|&(3.1)&\\
\end{tabular}

\end{center}
\end{footnotesize}
\caption{Example derivations.}
\label{fig:examplederivations}
\end{figure}

\end{document}